%% file: main.tex
\documentclass{article}
\usepackage[margin=2cm]{geometry}
\usepackage{graphicx}
\usepackage{booktabs}
\usepackage{tabularx}
\usepackage{geometry}
\usepackage{colortbl}
\usepackage{xcolor}
\usepackage{authblk}
\usepackage{natbib}
\usepackage{amsthm}
\newtheorem{proposition}{Proposition}

\newtheorem{lemma}{Lemma}

\usepackage{amsmath}
\usepackage{amssymb}
\usepackage{threeparttable}
\usepackage{booktabs}
\usepackage{algorithm}
\usepackage{algpseudocode}
\usepackage{tikz}
\usepackage{booktabs}
\usepackage{graphicx}
\usepackage{natbib}
 \bibpunct[, ]{(}{)}{,}{a}{}{,}%

\title{Portfolio Preference Elicitation in Institutional Crossing Markets}

\author{Yoontae Hwang\thanks{yoontae.hwang@pusan.ac.kr}}
\affil{Graduated School of Data Science, Pusan National University}
\date{May 2026}
\usepackage{amsthm}

\definecolor{pastelblue}{HTML}{4682B4} 

\usepackage{hyperref}
\hypersetup{
    colorlinks=true,
    citecolor=pastelblue, 
    linkcolor=black,
    urlcolor=blue
}

\usepackage{xcolor}

\definecolor{MSBlue}{HTML}{0072B2}    
\definecolor{MSOrange}{HTML}{E69F00}  
\definecolor{MSGreen}{HTML}{009E73}   
\definecolor{MSPurple}{HTML}{CC79A7}  
\definecolor{MSRed}{HTML}{D55E00}     
\definecolor{MSGray}{HTML}{4D4D4D}    

\begin{document}
\maketitle
\begin{abstract}
Institutional crossing platforms face a hidden-information problem: investors value trades as portfolios, but liquidity discovery is typically organized around individual securities. We model portfolio crossing as limited-communication preference elicitation over signed portfolio trades. The platform first uses price-directed demand queries to search the portfolio space and then verifies selected packages through value queries; an incumbent verification query records the demand-discovered allocation before further exploration. Final allocations are chosen from elicited reports, so the learning model guides queries but does not determine welfare. The analysis shows why search and verification are complementary. Demand queries locate high-value regions of a nonseparable portfolio space, but they provide only conservative welfare evidence unless selected packages are verified. Value queries provide exact welfare comparisons, but they are ineffective when applied to poorly targeted packages. Market-calibrated experiments using equity panels from the United States, Korea, Japan, and Germany show that demand-only and value-only designs recover only about half of full-information welfare under a limited query budget, whereas the hybrid procedure recovers 88\% and approaches 95\% as communication expands. We then compare exact security-level packages with factor-completed basket packages within the same allocation rule. Security-level packages are the unadjusted-efficiency mode when exact-securities disclosure is inexpensive. Factor-completed baskets become preferable when pretrade message informativeness is costly. The results characterize portfolio crossing as a selective verification problem and identify disclosure-sensitive package representation as a core design choice for hidden liquidity platforms.\end{abstract}

\input{sections/chap1_intro}


\input{sections/chap2_related}


\input{sections/chap3_method}

\input{sections/chap4_analysis} 

\input{sections/chap5_conclusion}

\bibliographystyle{plainnat}
\bibliography{main}

\clearpage
\appendix
\input{appendix/Proofs}
\input{appendix/detail}

\end{document}

%% file: sections/chap1_intro.tex
\section{Introduction}
\label{secIntro}

Institutional investors often trade portfolios rather than isolated stock orders. Empirical evidence on institutional trading treats a sequence of trades as an economically meaningful package, and intraquarter trading evidence shows that institutions actively adjust positions within short horizons \citep{chan1995behavior,puckett2011interim}. Hidden liquidity discovery, however, is usually organized around messages about individual securities. This creates a mismatch between the object being valued and the information unit being elicited. A trade that appears attractive in isolation may reduce welfare when it increases an existing exposure. Another leg may create value when it completes a hedge, relaxes a benchmark constraint, or offsets a liquidity cost. Transaction cost and dynamic trading models imply that optimal trades depend on current positions, trading targets, and the costs of moving across assets \citep{constantinides1986capital,garleanu2013dynamic}. The relevant object for hidden liquidity discovery is therefore a signed portfolio trade rather than a collection of separable single securities orders.

Block and upstairs markets allow large investors to obtain liquidity while limiting the exposure of large orders before execution \citep{seppi1990equilibrium,grossman1992informational}. Crossing networks and dark pools extend this logic to nondisplayed liquidity and venue choice \citep{hendershott2000crossing,zhu2014dark}. Market transparency can change liquidity and volatility, so message design is part of execution design \citep{madhavan1995consolidation,madhavan1996security}. Hidden order evidence further shows that traders adjust order exposure when execution probability and information leakage are both relevant \citep{bessembinder2009hidden,bloomfield2015hidden}. Portfolio trading creates a different design problem. The platform does not simply need to clear submitted orders. It must decide what limited portfolio level information to elicit before full demand is observed. The information architecture must recover enough nonseparable value information to identify mutually beneficial trades, while controlling the informativeness of pretrade messages.

We study this problem through query based portfolio crossing. Participants have private values over continuous signed equity portfolios. A demand query asks a participant to choose a trade at posted query prices. These queries identify active regions of portfolio demand. A value query asks for surplus on a selected package. These queries provide cardinal anchors for welfare comparisons. A bridge value query records the value of the interim allocation found after the demand query phase. The final allocation is selected from a finite report set. The surrogate model guides query selection, but final welfare is determined by elicited reports, demand query lower bounds, and a report based allocation rule.

The central design principle is that search and verification are different information tasks. Demand queries are useful because they use price directed responses to locate relevant regions of a high dimensional portfolio space. They are not sufficient for final welfare accounting because a demand response enters the allocation rule through a conservative lower bound unless it is verified. Value queries are useful because they turn selected packages into exact welfare candidates. They are not sufficient for search because exact values attached to poorly chosen packages do not reveal where mutually beneficial portfolio trades are likely to lie. Preference elicitation research provides the methodological basis for using selective queries rather than full valuation reports \citep{blum2004preference,lahaie2004applying}. Recent learning based auction designs further motivate using surrogate models to decide which packages should be queried under limited communication \citep{brero2018combinatorial,brero2019fast}. A hybrid protocol combines these two roles.

The procedure separates information architecture from package representation. We compare security level packages, denoted SL, with factor completed basket packages, denoted FC. SL packages use allocation relevant securities discovered from demand responses as the coordinates for later package queries. FC packages combine factor exposure targets with residual sleeves over the same demand discovered securities. Factor representations are natural in this setting because common exposures summarize economically relevant sources of portfolio risk \citep{ross1976arbitrage,fama1993common}. Empirical factor and shrinkage models support using a lower dimensional guide when the security universe is high dimensional \citep{connor1988risk,kozak2020shrinking}. Both package families use the same demand queries, bridge value query, value queries, surrogate search model, and report based allocation rule. They differ only in the package family used to form value query proposals.

The computational experiments use equity panels calibrated from the S\&P 500, KOSPI 200, Nikkei 225, and DAX. The results indicate that search and verification play distinct but complementary roles. Value only elicitation provides exact welfare information for the packages that are queried, but its performance depends on whether those packages are economically relevant. Demand only elicitation is more effective at locating relevant regions of the portfolio space, but demand reports enter the allocation rule through conservative lower bounds unless the associated packages are verified. Hybrid elicitation improves on both approaches because demand queries guide the search toward promising candidate packages, while value queries provide the cardinal information needed to compare those candidates in welfare terms. As the communication budget expands, this complementarity becomes more salient. The diagnostic experiment confirms the same mechanism. Demand queries can identify valuable candidate packages, but those packages support welfare comparisons only after selected candidates are verified.

The representation comparison gives a frontier rather than a universal ranking. Unadjusted efficiency refers to oracle relative welfare before accounting for disclosure costs. SL is the unadjusted efficiency representation when exact securities disclosure is inexpensive. FC has lower unadjusted efficiency, but it becomes the disclosure adjusted welfare frontier when pretrade message informativeness is penalized. In thin contra side liquidity regimes, same securities crossing weakens because direct opposing interest is scarce. Portfolio package representations can still create value by combining legs that offset each other through the portfolio curvature matrix.

This paper makes three contributions. First, it frames hidden liquidity discovery as a portfolio preference elicitation problem in which the relevant object is a signed portfolio trade rather than an individual order. Second, it develops a hybrid demand and value elicitation architecture for continuous equity portfolios. The analysis clarifies the distinct roles of demand reports, value reports, incumbent verification, and report based winner determination. Third, it treats package representation as an economic design variable. Security level packages and factor completed basket packages are compared within the same elicitation and allocation rule, which allows the analysis to separate unadjusted efficiency from disclosure adjusted welfare.

%% file: sections/chap2_related.tex
\section{Related Work}
\label{secRelated}

This paper connects four literatures. Market microstructure explains why institutions limit pretrade disclosure, but usually represents trading interest as a single order, a block, or a venue choice. Portfolio choice and execution explain why institutional trading values are nonseparable across securities, but usually take the desired portfolio trade as given. Combinatorial auction and preference elicitation research studies package allocation under limited communication, but mainly focuses on indivisible items. Work on portfolio auctions and basket trading shows that portfolio representation affects execution and information release. This paper combines these elements in a crossing procedure that elicits limited reports about continuous signed portfolio trades and treats package representation as a design variable.

\subsection{Hidden Liquidity and Financial Market Design}

For a crossing platform, the allocation problem is inseparable from the information problem. A large investor seeks liquidity, but revealing trading interest can expose private information and degrade execution quality. Classical market microstructure models show that order flow conveys information and that liquidity is affected by what traders reveal before execution \citep{kyle1985continuous,admati1991sunshine}. Models of block and upstairs trading show that hidden or intermediated trading can reduce the exposure cost of large orders \citep{seppi1990equilibrium,grossman1992informational}.

Transparency is not mechanically welfare improving in these settings. Fragmented and less transparent markets can reduce the cost of executing large orders when traders face repeated exposure of their intentions \citep{madhavan1995consolidation,madhavan1996security}. Empirical work on institutional executions shows that execution costs depend on order size, investment style, and the trading mechanism used to access liquidity \citep{keim1995anatomy,chan1997institutional}. Evidence from upstairs markets further shows that negotiated and intermediated mechanisms can support large block trades while interacting with displayed markets \citep{keim1996upstairs,madhavan1997search}.

A separate empirical literature studies nondisplayed liquidity directly. Hidden and reserve order evidence shows that traders choose exposure jointly with price, quantity, and execution risk \citep{bessembinder2009hidden,buti2013undisclosed}. Experimental and venue sorting evidence shows that dark trading changes order submission and venue selection when traders face different combinations of immediacy and information leakage \citep{bloomfield2015hidden,menkveld2017shades}. Work on dark pools and fragmentation studies how nondisplayed trading interacts with displayed markets and price discovery \citep{degryse2015impact,zhu2014dark}. Recent evidence on broker networks highlights the continuing role of intermediated search in institutional liquidity provision \citep{buti2022diving,han2024institutional}.

A related market design literature studies how trading protocols shape liquidity and market clearing. Frequent batch auctions restructure the timing of exchange trading to address speed competition \citep{budish2015hft}. Flow Trading allows participants to submit user defined linear combinations of assets and clears the resulting portfolio demand in discrete time \citep{budish2023flow}. The present paper studies a different margin. Flow Trading asks how a market should clear submitted portfolio demand. Query based portfolio crossing asks what portfolio information a hidden liquidity platform should elicit before full demand is observed. The contribution is therefore an information architecture for selective portfolio preference elicitation rather than a clearing rule for fully submitted portfolio orders.

\subsection{Portfolio Trading and Multi Asset Execution}

Models of portfolio execution typically solve an individual trader's optimization problem given a fixed objective, target position, or predefined basket. A crossing platform faces a multi agent inference problem. It must learn enough about several private portfolio objectives to identify mutually beneficial combinations. Portfolio choice theory establishes why this inference is nonseparable because a trade changes expected return, risk, benchmark deviation, and constraint slack at the portfolio level \citep{markowitz1952portfolio,grinold1999active}. Factor based asset pricing gives an economic basis for summarizing portfolio risk through common exposures rather than through securities independently \citep{ross1976arbitrage,fama1993common}.

Empirical factor models and shrinkage methods reinforce the value of low dimensional representations in high dimensional equity settings \citep{connor1988risk,kozak2020shrinking}. Multi asset execution models generate related nonseparability because the value of one leg depends on the other legs in the basket and on the path of execution \citep{almgren2001optimal,bertsimas1998optimal}. Portfolio optimization with transaction costs and constraints gives the same implication for rebalancing because trades are chosen relative to the current portfolio and the feasible region \citep{constantinides1986capital,lobo2007portfolio}. Dynamic trading models show that partial adjustment is optimal when expected returns are predictable and trading is costly \citep{garleanu2013dynamic}.

Institutional trading behavior reinforces this portfolio view. Large investors do not only decide whether to buy or sell a standalone stock. Their trading decisions reflect package size, investment style, transaction costs, hedge completion, and rebalancing needs \citep{chan1995behavior,puckett2011interim}. A purchase can be unattractive when it increases an existing risk exposure. A sale can be valuable when it reduces benchmark deviation or completes a hedge. Evidence on downward sloping demand curves also suggests that securities are not perfect substitutes in the short run \citep{shleifer1986demand,wurgler2002arbitrage}. The appropriate unit for liquidity discovery is therefore a signed portfolio trade. This turns portfolio trading into an information elicitation problem because the platform observes only selected responses to selected queries.

\subsection{Preference Elicitation and Combinatorial Auctions}

Combinatorial auction theory studies allocation when valuations are nonseparable across items. Package bidding allows participants to express synergies, but full preference revelation imposes communication and computation costs \citep{ausubel2002ascending,cramton2006combinatorial}. Efficient allocation can require substantial communication when preferences are unrestricted, and practical auction languages therefore impose structure on reports \citep{nisansegal2006communication,bichler2023fuel}. Preference elicitation addresses this problem by using selective queries to approximate efficient allocations without exhaustive valuation reports \citep{conen2001preference,blum2004preference}.

The query design literature distinguishes value queries, demand queries, and learning based methods for selecting which information to ask. Learning based elicitation uses observed responses to focus later queries on allocation relevant regions \citep{lahaie2004applying,brero2018combinatorial}. Bayesian and machine learning based iterative auctions show that statistical models can improve query selection while final allocations can still be based on reported information \citep{brero2019fast,soumalias2025prices}. This distinction is central to the present paper because the surrogate model guides query selection, but the final allocation is selected from elicited reports and valid lower bounds.

The financial crossing setting differs from standard combinatorial auction domains in several respects. Packages are signed trade vectors rather than discrete bundles. Feasibility includes gross exposure limits, securities level bounds, and residual external execution costs. Messages have pretrade information content because query responses may reveal alpha, inventory pressure, benchmark demand, or hedging motives. This information content links preference elicitation to market microstructure rather than to communication cost alone \citep{kyle1985continuous,madhavan1996security}. It also creates scope for strategic exposure costs when other traders can condition on revealed trading needs \citep{brunnermeier2005predatory}.

\subsection{Portfolio Representation, Baskets, and Information Release}

Research on portfolio auctions and basket trading shows that a basket is not a neutral container. Its composition determines what risk is transferred, how liquidity demand reaches underlying securities, and what information is revealed during execution. Blind portfolio auctions highlight how intermediaries manage information release during large portfolio sales \citep{padilla2012intermediated}. ETF and index arbitrage research shows that creation and redemption baskets transmit liquidity demand between fund shares and underlying securities \citep{bendavid2018etfs,pan2019etf}. ETF activity can also affect price discovery, liquidity, and return comovement in underlying securities \citep{glosten2021etf,madhavan2016price}. Evidence on ETF ownership and arbitrage shows that basket based trading can change information production and asset return correlations \citep{da2018exchange,israeli2017dark}.

This literature motivates the representation margin in the crossing platform's design. Security level packages specify exact securities and quantities. Factor completed basket packages combine factor exposure targets with residual completion over demand discovered securities. The factor component builds on the idea that common exposures summarize priced and economically relevant risk \citep{ross1976arbitrage,fama1993common}. The residual sleeve preserves participant specific active demand that may not be captured by the factor guide \citep{connor1988risk,kozak2020shrinking}. Existing research often treats portfolio representation as a fixed trading object or an information release device. This paper treats it as a design variable in the elicitation protocol.

%% file: sections/chap3_method.tex
\section{Model and Elicitation Architecture}
\label{secMethod}

We model institutional crossing as a signed multi asset allocation problem. Participant \(i\)'s trade is a signed vector \(d_i\in\mathbb R^m\), where positive coordinates represent purchases and negative coordinates represent sales. A package is a queried or reported trade vector. An allocation is a collection of packages assigned across participants. The portfolio level formulation is necessary because risk, benchmark deviation, liquidity exposure, and hedge completion create complementarities across trade legs.

Let \(\mathcal J=\{1,\ldots,m\}\) denote the stock universe and let \(\mathcal I=\{1,\ldots,n\}\) denote the participant set. Feasible trades are restricted to a convex set \(\mathcal X_i\subset\mathbb R^m\) that contains the no trade vector. This set captures gross exposure limits, securities level bounds, mandate restrictions, and short sale constraints. A pure internal crossing allocation satisfies \(\sum_i d_i=0\). When internal crossing leaves an imbalance, the residual vector is \(\xi=-\sum_i d_i\), and the platform incurs residual external execution cost \(\Psi(\xi)\). The baseline specification is \(\Psi(\xi)=\xi^\top\Gamma\xi/2\), with \(\Gamma\succ0\).

Participant \(i\)'s private surplus is the net improvement in a risk adjusted trading objective. Let \(\tau_i=b_i-h_i\) denote the distance between target holdings and current holdings. Let \(\alpha_i\) represent alpha, urgency, or execution needs. Let \(H_i\succeq0\) capture curvature from risk, liquidity, and mandates. The primitive surplus from trade \(d\) is \(\alpha_i^\top d-(d-\tau_i)^\top H_i(d-\tau_i)/2+\tau_i^\top H_i\tau_i/2\). This yields the reduced form valuation
\begin{equation}
 v_i(d)=\theta_i^\top d-\frac{1}{2}d^\top H_i d,
 \qquad
 \theta_i=\alpha_i+H_i\tau_i,
 \qquad
 H_i=\lambda_i\Sigma+\gamma_i\Delta+\rho_i I .
\label{eqValueModel}
\end{equation}
Here \(\Sigma\) is the return covariance matrix, \(\Delta\) is a diagonal liquidity or idiosyncratic cost matrix, and \(I\) provides residual mandate curvature. 

The full information benchmark is the welfare maximizing allocation observed by an omniscient platform
\begin{equation}
W^\star
=
\max_{\{d_i\in\mathcal X_i\}_{i\in\mathcal I}}
\left\{
\sum_{i\in\mathcal I}v_i(d_i)
-
\Psi\left(-\sum_{i\in\mathcal I}d_i\right)
\right\}.
\label{eqOracle}
\end{equation}
For any realized allocation \(\hat d\), oracle relative welfare efficiency is
\begin{equation}
\operatorname{Eff}(\hat d)=\frac{\sum_{i\in\mathcal I} v_i(\hat d_i)-\Psi\left(-\sum_{i\in\mathcal I}\hat d_i\right)}{W^\star}.
\label{eqEfficiency}
\end{equation}
Because the no trade allocation is feasible and yields zero welfare, \(W^\star\ge0\).

\begin{proposition}[Portfolio complementarity]
\label{propPortfolioComb}
Suppose \(v_i(d)=\theta_i^\top d-d^\top H_i d/2\) and \(H_i\succeq0\). For any two trade increments \(g\) and \(h\), the joint surplus satisfies \(v_i(g+h)-v_i(g)-v_i(h)=-g^\top H_i h\). The increments are complementary for participant \(i\) when \(g^\top H_i h<0\) and substitutable when \(g^\top H_i h>0\).
\end{proposition}

Proposition \ref{propPortfolioComb} provides the portfolio accounting identity behind nonseparable package values. The result shows that complementarity does not require indivisible assets. It follows from cross terms in portfolio risk and execution curvature. Same direction trades in correlated securities tend to be substitutable when they concentrate exposure. Long short packages can be complementary when one leg offsets the risk or cost exposure created by another leg.

\subsection{Package Representations}
\label{secQueryFamilyConstruction}

The representation step determines how query objects are converted into feasible signed trades. This step is separate from the allocation rule. Security level packages use demand discovered securities as package coordinates. Factor completed basket packages combine factor exposure targets with residual completion over the same demand discovered securities. This separation allows SL and FC packages to enter the same report based allocation problem while differing only in the representation used to generate value query proposals.

The elicitation process begins with demand queries that identify participant specific allocation relevant securities. These are securities whose absolute demand query quantities exceed a materiality threshold or rank among the participant's largest reported positions. Let \(S_i\subset\mathcal J\) denote this set, and let \(E_{S_i}\) embed the selected coordinates into the full stock universe.

Security level packages use these allocation relevant securities directly. The query family is
\begin{equation}
\mathcal Q_i^{\mathrm{SL}}
=
\left\{
E_{S_i}w
\mid
w\in\mathcal W_i^{\mathrm{SL}}
\right\},
\qquad
\mathcal W_i^{\mathrm{SL}}
\subseteq
\left\{
 w
 \mid
 E_{S_i}w\in\mathcal X_i
\right\}.
\label{eqSecurityLevelFamily}
\end{equation}
This representation retains exact security level expressiveness in the value query phase.

Factor completed basket packages construct packages through a lower dimensional factor representation with residual completion. Representing portfolios through common factors is consistent with arbitrage pricing theory and empirical factor models \citep{ross1976arbitrage,connor1988risk}. Using a lower dimensional guide is also consistent with modern cross sectional shrinkage methods for large equity panels \citep{fama1993common,kozak2020shrinking}. Let the covariance matrix be approximated by \(\Sigma\approx F\Omega F^\top+\Delta\), where \(F\in\mathbb R^{m\times k}\) and \(k\ll m\). Let \(K=\Sigma+\nu\Delta+\rho I\) be a positive definite metric. Given a factor exposure target \(\eta\in\mathbb R^k\), the minimum cost factor matching portfolio solves \(\min_a a^\top K a/2\) subject to \(F^\top a=\eta\), with solution
\begin{equation}
a(\eta)=A\eta, \quad A=K^{-1}F(F^\top K^{-1}F)^{-1}.
\label{eqAtomMatrix}
\end{equation}
A security level trade \(u=E_{S_i}w\) may introduce unintended factor exposure.
The factor completion map finds the closest portfolio under metric \(K\) that matches
the target exposure \(\eta\):
\begin{equation}
C(u,\eta)=A\eta+(I-AF^\top)u, \quad R_K=I-AF^\top .
\label{eqCompletion}
\end{equation}
The factor completed query family is
\begin{equation}
\mathcal Q_i^{\mathrm{FC}}
=
\left\{
A\eta
+
R_K E_{S_i}w
\mid
(\eta,w)\in\mathcal Y_i^{\mathrm{FC}}
\right\},
\qquad
\mathcal Q_i^{\mathrm{FC}}\subseteq\mathcal X_i .
\label{eqFactorCompletedFamily}
\end{equation}
For polyhedral feasibility constraints \(\mathcal X_i=\{d\in\mathbb R^m\mid L_i d\le r_i\}\), an admissible set is \(\mathcal Y_i^{\mathrm{FC}}=\{(\eta,w)\mid L_i(A\eta+R_KE_{S_i}w)\le r_i,\ \eta\in\mathcal Z_i,\ w\in\mathcal W_i^{\mathrm{FC}}\}\). For general convex constraints, \(\mathcal Y_i^{\mathrm{FC}}\) is the preimage of \(\mathcal X_i\) under the factor completion map.

The two query families define different package representations for the same report based allocation problem. SL packages use demand discovered securities as package coordinates. FC packages augment those coordinates with factor exposure targets and the completion map \(R_K\). Demand queries, the bridge value query, value reports, and the final winner determination problem are common across SL and FC.

\subsection{Demand Queries and Surrogate Search}
\label{secQueryEstimation}

A demand query presents a hypothetical price vector \(p^\ell\in\mathbb R^m\). Under the signed trade convention, \(p_j^\ell\) is a per share shadow execution cost for stock \(j\). A higher \(p_j^\ell\) reduces the attractiveness of positive demand in securities \(j\) and increases the attractiveness of negative demand. Participant \(i\) responds by choosing
\begin{equation}
 d_i^\ell
 \in
 \arg\max_{d\in\mathcal X_i}
 \left\{
 v_i(d)
 -
 {p^\ell}^\top d
 \right\}.
\label{eqDQ}
\end{equation}
This response implies the revealed preference inequality \(v_i(d_i^\ell)-{p^\ell}^\top d_i^\ell\ge v_i(d)-{p^\ell}^\top d\) for all \(d\in\mathcal X_i\). Under the quadratic model, an interior response satisfies \(\theta_i-H_i d_i^\ell=p^\ell\). With active constraints, the condition becomes the variational inequality \((\theta_i-H_i d_i^\ell-p^\ell)^\top(d-d_i^\ell)\le0\) for all \(d\in\mathcal X_i\).

Query prices are restricted to a low dimensional basis that evolves across demand query rounds. Initially, the platform uses \(\Phi^0=[F,G]\), where \(G\) is a baseline liquidity component. As demand responses reveal allocation relevant securities, the basis is augmented with those
security coordinates:
\begin{equation}
p^\ell=\Phi^\ell\kappa^\ell, \quad \Phi^\ell=[F,G,E_{C^\ell}].
\label{eqTimeVaryingPriceBasis}
\end{equation}
The predicted crossing problem with surrogate valuations \(\hat v_i\) and residual execution has dual objective \(\mathfrak C(p)=\sum_i U_i(p)+\Psi^\ast(-p)\), where \(U_i(p)=\max_{d\in\mathcal X_i}\{\hat v_i(d)-p^\top d\}\). When \(\Psi(\xi)=\xi^\top\Gamma\xi/2\), any predicted optimal demand \(\hat d_i(p)\) gives the subgradient \(-\sum_i\hat d_i(p)+\Gamma^{-1}p\). The projected subgradient step over prices of the form \(p=\Phi^\ell\kappa\) is
\begin{equation}
\kappa^{\ell+1} = \kappa^\ell + \eta_\ell {\Phi^\ell}^{\top} \left[ \sum_i\hat d_i(\Phi^\ell\kappa^\ell) - \Gamma^{-1}\Phi^\ell\kappa^\ell \right].
\label{eqPriceUpdate}
\end{equation}
The projected dual prices are query prices rather than final market clearing prices for a fully submitted order book. They guide participants toward portfolio regions that are likely to be allocation relevant, while the implemented allocation remains report based. The platform maintains a structural surrogate valuation model \(\hat v_i(d,\vartheta_i)=\beta_i^\top d-\lambda_i d^\top\Sigma d/2-\gamma_i d^\top\Delta d/2-\rho_i\|d\|_2^2/2\), with \(\lambda_i,\gamma_i,\rho_i\ge0\), and parameter vector \(\vartheta_i=(\beta_i,\lambda_i,\gamma_i,\rho_i)\). Demand queries supply projected marginal moments. If \(\mathcal X_i=\{d\in\mathbb R^m\mid L_i d\le r_i\}\), let \(\mathcal A_i^\ell\) be the active constraint set at \(d_i^\ell\), let \(\mathcal T_i^\ell=\{s\in\mathbb R^m\mid (L_i)_{\mathcal A_i^\ell}s=0\}\), and let \(P_i^\ell\) be the orthogonal projection onto \(\mathcal T_i^\ell\). The projected moment is \(P_i^\ell[\beta_i-\lambda_i\Sigma d_i^\ell-\gamma_i\Delta d_i^\ell-\rho_i d_i^\ell-p^\ell]\approx0\).

Value queries provide cardinal surplus observations. If package \(q_i^\ell\) is queried and value \(\bar v_i^\ell=v_i(q_i^\ell)\) is reported, the estimator solves
\begin{equation}
\begin{aligned}
\mathcal L_i^{\mathrm{dq}}(\vartheta_i)
&=
\sum_{\ell\in\mathcal R_i^{\mathrm{dq}}}
\left\|
P_i^\ell
\left[
\beta_i
-
\lambda_i\Sigma d_i^\ell
-
\gamma_i\Delta d_i^\ell
-
\rho_i d_i^\ell
-
p^\ell
\right]
\right\|_2^2,\\
\mathcal L_i^{\mathrm{vq}}(\vartheta_i)
&=
\sum_{\ell\in\mathcal R_i^{\mathrm{vq}}}
\left[\hat v_i(q_i^\ell,\vartheta_i)-\bar v_i^\ell\right]^2,\\
\min_{\vartheta_i\,\mid\,\lambda_i,\gamma_i,\rho_i\ge0}
\quad
&
\mathcal L_i^{\mathrm{dq}}(\vartheta_i)
+
\omega_{\mathrm{vq}}\mathcal L_i^{\mathrm{vq}}(\vartheta_i)
+
\omega_{\mathrm{reg}}\|\vartheta_i\|_2^2 .
\end{aligned}
\label{eqMixedLearning}
\end{equation}
The estimator is used to select subsequent queries. Final welfare is evaluated from the elicited report set.

\subsection{Report-Based Allocation and Welfare Certification}
\label{secBridge}

The final allocation is selected from explicitly elicited reports. The no trade package is always available with value zero. A demand query response satisfies revealed preference against no trade and therefore provides a conservative value bound. Each demand query report is recorded with a lower bound \(\underline v_i^\ell\le v_i(d_i^\ell)\). Value query reports record surplus for selected packages. Let \(\mathcal R_i\) collect participant \(i\)'s reports and let \(\mathcal E_i(\mathcal R_i)\) be the finite set of reported packages together with no trade. The inferred report value is
\begin{equation}
\tilde v_i(q,\mathcal R_i)
=
\begin{cases}
 v_i(q) & \text{if }(q,v_i(q))\in\mathcal R_i^{\mathrm{vq}},\\
 \max\{\underline v_i^\ell\mid q=d_i^\ell,\ (d_i^\ell,p^\ell)\in\mathcal R_i^{\mathrm{dq}}\} & \text{if }q\text{ appeared in a demand query},\\
 0 & \text{if }q=0,\\
 -\infty & \text{otherwise.}
\end{cases}
\label{eqInferredValue}
\end{equation}
The final allocation solves the report based winner determination problem
\begin{equation}
a^\star(\mathcal R)
\in
\arg\max_{\{q_i\in\mathcal E_i(\mathcal R_i)\}_{i\in\mathcal I}}
\left\{
\sum_i\tilde v_i(q_i,\mathcal R_i)
-
\Psi\left(-\sum_i q_i\right)
\right\}.
\label{eqReportWDP}
\end{equation}
When \(\Psi(\xi)=\xi^\top\Gamma\xi/2\), the finite report based problem can be written with binary package selection variables. Let \(\mathcal E_i(\mathcal R_i)=\{q_{ir}\mid r\in\mathcal K_i\}\), let \(y_{ir}\in\{0,1\}\) indicate the selected package, and let \(\tilde v_{ir}=\tilde v_i(q_{ir},\mathcal R_i)\). The finite problem is
\begin{align}
\max_{\{y_{ir}\}}
\quad
&
\sum_{i\in\mathcal I}\sum_{r\in\mathcal K_i}\tilde v_{ir}y_{ir}
-
\frac{1}{2}
\left(
\sum_{i\in\mathcal I}\sum_{r\in\mathcal K_i}q_{ir}y_{ir}
\right)^\top
\Gamma
\left(
\sum_{i\in\mathcal I}\sum_{r\in\mathcal K_i}q_{ir}y_{ir}
\right)\\
\text{s.t.}
\quad
&
\sum_{r\in\mathcal K_i}y_{ir}=1,
\qquad
 i\in\mathcal I,\\
&
 y_{ir}\in\{0,1\},
\qquad
 i\in\mathcal I,
 r\in\mathcal K_i .
\end{align}
The first welfare certification step is an accounting result. Demand query reports are useful for search, but they enter the final allocation problem through valid lower bounds unless they are verified. Value query reports turn selected packages into exact welfare candidates. The bridge value query applies this accounting structure to the interim allocation found after the demand query phase.

For an allocation \(a=(q_i)_{i\in\mathcal I}\) selected from the finite report set, define true welfare by \(W(a)=\sum_{i\in\mathcal I} v_i(q_i)-\Psi(-\sum_{i\in\mathcal I}q_i)\), and define reported welfare by \(\widetilde W(a\mid\mathcal R)=\sum_{i\in\mathcal I}\tilde v_i(q_i,\mathcal R_i)-\Psi(-\sum_{i\in\mathcal I}q_i)\). Let \(\mathcal A(\mathcal R)=\prod_{i\in\mathcal I}\mathcal E_i(\mathcal R_i)\) be the finite report allocation set. Let \(\mathcal A^{\mathrm{exact}}(\mathcal R)\subseteq\mathcal A(\mathcal R)\) denote the set of allocations for which every selected participant package is exact valued, with the no trade package counted as exact valued with value zero.

\begin{lemma}[Exact-valued candidate preservation]
\label{lemExactCandidatePreservation}
Assume that every inferred report value is a valid lower bound on the participant's true value. Let \(a^{\mathrm{final}}\in\arg\max_{a\in\mathcal A(\mathcal R)}\widetilde W(a\mid\mathcal R)\). Then \(W(a^{\mathrm{final}})\ge\max_{a\in\mathcal A^{\mathrm{exact}}(\mathcal R)} W(a)\). If \(a^{\mathrm{dq}}\) is the interim allocation after the demand query phase and the bridge value query records \(v_i(a_i^{\mathrm{dq}})\) for every selected package, then subsequent report based allocation satisfies \(W(a^{\mathrm{final}})\ge W(a^{\mathrm{dq}})\).
\end{lemma}

Lemma \ref{lemExactCandidatePreservation} gives the accounting role of the bridge value query. The demand query phase may identify a high welfare allocation, but conservative lower bound accounting may undervalue it. The bridge value query records exact surplus for the selected packages. Subsequent value query exploration can then enlarge the report set without removing the verified incumbent from the set of exact valued welfare candidates.

The next result adds the search component. The surrogate dual objective is used to choose query prices. To state the welfare certificate against the full information benchmark, define the corresponding true value dual by \(U_i^v(p)=\max_{d\in\mathcal X_i}\{v_i(d)-p^\top d\}\) and \(\mathfrak C^v(p)=\sum_{i\in\mathcal I}U_i^v(p)+\Psi^*(-p)\), where \(\Psi^*(z)=\sup_{\xi\in\mathbb R^m}\{z^\top\xi-\Psi(\xi)\}\) is the convex conjugate of the residual execution cost. The superscript on \(U_i^v\) emphasizes that this is the economic dual formed from the true valuation \(v_i\), rather than the surrogate dual used for query selection.

\begin{proposition}[Demand-query dual certificate]
\label{propDemandDualCertificate}
Assume that \(\Psi\) is proper, closed, and convex and that each \(\mathcal X_i\) is nonempty. For any price vector \(p\in\mathbb R^m\) such that \(U_i^v(p)\) and \(\Psi^*(-p)\) are finite, \(W^\star\le \mathfrak C^v(p)\). Let \(a\in\mathcal A^{\mathrm{exact}}(\mathcal R)\) be any exact valued allocation available in the final report set. Under valid lower bound accounting,
\begin{equation}
W^\star-W(a^{\mathrm{final}})
\le \mathfrak C^v(p)-W(a) \quad \text{for every }p\in\mathbb R^m .
\label{eqGeneralDualCertificate}
\end{equation}
In particular, after the bridge value query verifies the demand discovered interim allocation \(a^{\mathrm{dq}}\),
\begin{equation}
W^\star-W(a^{\mathrm{final}})
\le
\mathfrak C^v(p)-W(a^{\mathrm{dq}})
\qquad
\text{for every }p\in\mathbb R^m .
\label{eqBridgeDualCertificate}
\end{equation}
Suppose further that \(d_i(p)\in\arg\max_{d\in\mathcal X_i}\{v_i(d)-p^\top d\}\) is the participant's demand response at a common query price \(p\), let \(d(p)=(d_i(p))_{i\in\mathcal I}\), and define \(\xi(p)=-\sum_{i\in\mathcal I}d_i(p)\). If the bridge value query records exact values for the package profile \(d(p)\), then
\begin{equation}
W^\star-W(a^{\mathrm{final}}) \le \mathfrak C^v(p)-W(d(p)) = \Psi(\xi(p))+\Psi^*(-p)+p^\top\xi(p).
\label{eqDQDualGapGeneral}
\end{equation}
If \(\Psi(\xi)=\xi^\top\Gamma\xi/2\) with \(\Gamma\succ0\), this bound becomes
\begin{equation}
W^\star-W(a^{\mathrm{final}})
\le
\frac{1}{2}
\left(
-\sum_{i\in\mathcal I}d_i(p)+\Gamma^{-1}p
\right)^\top
\Gamma
\left(
-\sum_{i\in\mathcal I}d_i(p)+\Gamma^{-1}p
\right).
\label{eqDQDualGapQuadratic}
\end{equation}
\end{proposition}

Proposition \ref{propDemandDualCertificate} formalizes the search role of demand queries. A query price \(p\) generates a dual upper bound on the full information crossing benchmark. A verified package profile generates a welfare lower bound that is preserved by Lemma \ref{lemExactCandidatePreservation}. Their difference is a certificate of the final allocation's welfare loss. In the implemented algorithm, the bridge value query verifies \(a^{\mathrm{dq}}\), so \eqref{eqBridgeDualCertificate} is the operative certificate. The closed form residual gap in \eqref{eqDQDualGapQuadratic} applies to the special case in which the verified package profile is the simultaneous demand response at price \(p\). Thus demand queries provide price directed search, while value queries provide the accounting step needed to convert discovered packages into certified welfare candidates.

After the bridge round, the platform issues model guided value queries from the predicted crossing problem over the selected query family. Let \(\chi\in\{\mathrm{SL},\mathrm{FC}\}\) denote the representation. The predicted allocation is
\begin{equation}
\hat a^\ell
\in
\arg\max_{\{q_i\in\mathcal Q_i^{\chi}\}_{i\in\mathcal I}}
\left\{
\sum_i \hat v_i(q_i)
-
\Psi\left(-\sum_iq_i\right)
\right\}.
\label{eqPredictedVQAllocation}
\end{equation}
Participants report \(v_i(\hat a_i^\ell)\) for these model guided proposals. The allocation rule remains the same under SL and FC. The selected representation determines only the package family used for value query proposals.

\begin{algorithm}
\caption{Query Based Portfolio Crossing Procedure}
\label{algDemandDiscovered}
\begin{algorithmic}[1]
\State Estimate \(\Sigma\), \(F\), \(\Delta\), \(\Gamma\), and \(K=\Sigma+\nu\Delta+\rho I\) from market data
\State Select package family \(\chi\in\{\mathrm{SL},\mathrm{FC}\}\)
\State Initialize demand query reports \(\mathcal R^{\mathrm{dq}}\), value query reports \(\mathcal R^{\mathrm{vq}}\), and allocation relevant securities set \(C^1=\emptyset\)
\For{\(\ell=1,\ldots,Q_{\mathrm{dq}}\)}
    \State Form \(\Phi^\ell=[F,G,E_{C^\ell}]\) and set \(p^\ell=\Phi^\ell\kappa^\ell\)
    \For{each participant \(i\)}
        \State Receive \(d_i^\ell\in\arg\max_{d\in\mathcal X_i}\{v_i(d)-{p^\ell}^\top d\}\)
        \State Add \((d_i^\ell,p^\ell)\) to \(\mathcal R_i^{\mathrm{dq}}\)
    \EndFor
    \State Estimate \(\hat v_i(\cdot,\vartheta_i)\) using \eqref{eqMixedLearning}
    \State Update \(\kappa^{\ell+1}\) using the projected subgradient step described above
    \State Update \(C^{\ell+1}\) from demand query responses observed through round \(\ell\)
\EndFor
\State Compute \(a^{\mathrm{dq}}\) from \eqref{eqReportWDP}
\For{each participant \(i\)}
    \State Ask the bridge value query \(v_i(a_i^{\mathrm{dq}})\)
    \State Add \((a_i^{\mathrm{dq}},v_i(a_i^{\mathrm{dq}}))\) to \(\mathcal R_i^{\mathrm{vq}}\)
\EndFor
\For{\(\ell=1,\ldots,Q_{\mathrm{vq}}-1\)}
    \State Extract participant specific allocation relevant securities sets \(S_i\) from demand query reports
    \State Construct \(\mathcal Q_i^\chi\) using \eqref{eqSecurityLevelFamily} or \eqref{eqFactorCompletedFamily}
    \State Estimate \(\hat v_i(\cdot,\vartheta_i)\) using all reports
    \State Solve \eqref{eqPredictedVQAllocation} over \(\mathcal Q_i^\chi\)
    \For{each participant \(i\)}
        \State Ask \(v_i(\hat a_i^\ell)\)
        \State Add \((\hat a_i^\ell,v_i(\hat a_i^\ell))\) to \(\mathcal R_i^{\mathrm{vq}}\)
    \EndFor
\EndFor
\State Return \(a^{\mathrm{final}}\) from \eqref{eqReportWDP} using all reports
\end{algorithmic}
\end{algorithm}

%% file: sections/chap4_analysis.tex
\section{Calibrated Market Design Experiments}
\label{secExperiments}

The computational experiments evaluate query based portfolio crossing in calibrated equity market environments. The objective is to measure how much welfare a disclosure sensitive platform can recover when it observes selected reports about signed portfolio trades. The primary outcome is oracle relative welfare efficiency, reported as \(100\times\operatorname{securities}{Eff}\), where \(\operatorname{securities}{Eff}\) is defined in \eqref{eqEfficiency}. The benchmark is the full information allocation in \eqref{eqOracle}. The implemented allocation is chosen from the finite report set through \eqref{eqReportWDP}. The surrogate estimator in \eqref{eqMixedLearning} guides query selection and does not directly determine final welfare.

The experiments are organized around two design questions. The first question concerns information architecture. Demand queries are price directed questions of the form \eqref{eqDQ} that reveal locally optimal portfolio trades and identify active regions of demand. Value queries ask for surplus on selected packages and provide cardinal welfare comparisons for the report based allocation rule. The bridge value query records the demand discovered allocation with exact values and implements the accounting logic in Lemma \ref{lemExactCandidatePreservation}. The second question concerns package representation. SL uses demand discovered securities as in \eqref{eqSecurityLevelFamily}. FC uses the factor completion map in \eqref{eqCompletion} and the query family in \eqref{eqFactorCompletedFamily}. Oracle relative efficiency before accounting for disclosure costs is referred to as unadjusted efficiency. SL is the exact securities mode when pretrade securities disclosure is inexpensive. FC is the basket mode that reduces message informativeness when pretrade disclosure is costly.

\subsection{Data, Calibration, and Inference}
\label{secExpData}

This subsection fixes the empirical environment used in the main comparisons. The purpose is to make comparisons across elicitation architectures and representations occur within the same market instances rather than across different simulated economies. The main parameters are held fixed in the baseline runs so that changes in outcomes can be attributed to the protocol being compared. Appendix Tables \ref{tabAppParticipantCalibration},report the calibration and query design details.

The empirical inputs are constructed from four equity panels, the S\&P 500, KOSPI 200, Nikkei 225, and DAX. The U.S. panel uses CRSP daily stock files. The non U.S. panels use Compustat Global daily security files, with returns measured in local currency. Each market uses the most recent 252 trading days ending on a date with nondegenerate cross sectional return dispersion. Securities must satisfy an 80\% return coverage screen, and the retained universe contains the largest eligible securities by terminal market capitalization. The retained universe and factor dimension are fixed before protocol comparisons are run. We use the largest eligible securities and the leading covariance factors to create a common institutional trading environment across the four panels. Holding this environment fixed ensures that differences in outcomes come from the elicitation architecture or package representation, rather than from changes in the market inputs supplied to the platform.

The covariance input is estimated from winsorized returns, shrunk toward the diagonal, and projected onto the positive semidefinite cone. Liquidity curvature is \(\Delta=\operatorname{diag}(1/\ell_j)\), where \(\ell_j\) is terminal market capitalization divided by the market median and truncated below at 0.05. Factor exposures are the leading principal components of the covariance matrix. The atom matrix \(A\) and residual projection \(R_K\) are constructed as in \eqref{eqAtomMatrix} and \eqref{eqCompletion}. These fixed inputs are used by both SL and FC, which ensures that the representation comparison isolates package design rather than differences in the market data supplied to the platform.

Each matched cell contains an institutional liquidity pool spanning five trading motives: benchmark rebalancing, active alpha demand, long-short residual demand, basket trading, and dealer inventory pressure. The same participant profile sequence and private-value shocks are then used across protocols within the cell. The profile determines covariance curvature, liquidity sensitivity, gross exposure, and securities level limits. Private values are generated from \eqref{eqValueModel}. Target trades combine a common factor component, a residual securities component, and profile specific alpha or inventory pressure. Contra side liquidity, denoted \(\zeta_{\mathrm{contra}}\), governs how much direct opposing interest exists in the market. Low values move participants toward a common side. High values create richer offsetting interest. The main architecture experiments keep this environment fixed, while Table \ref{tabContraSideLiquidity} varies the contra side liquidity regime.
\begin{figure}[t]
\centering
\includegraphics[width=0.96\linewidth]{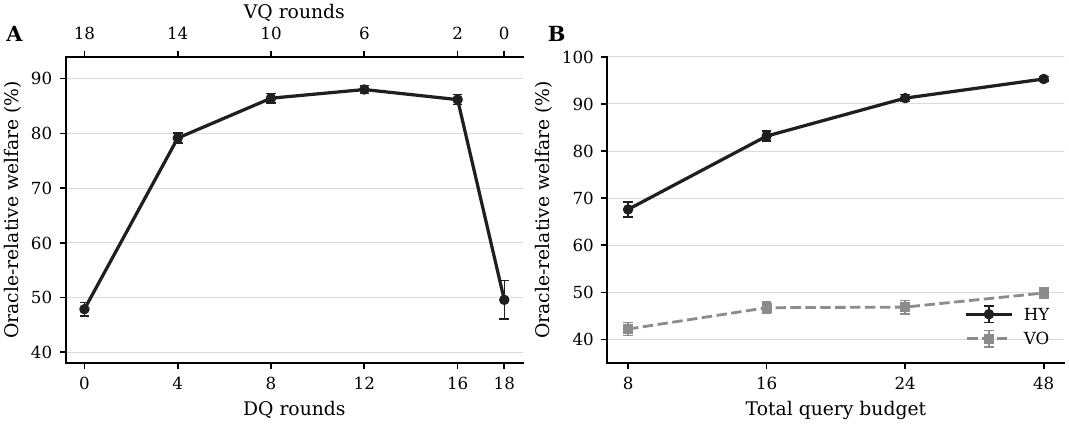}
\caption{Search and verification architecture under limited communication}
\label{figMainArchitecture}
\vspace{0.25em}
\begin{minipage}{0.92\linewidth}
\footnotesize
Notes. The left panel fixes the total budget at 18 queries and varies the demand query rounds before switching to value queries. The right panel varies the total query budget and compares HY with VO. Vertical intervals are stratified bootstrap 95\% half widths. HY denotes hybrid demand and value elicitation. VO denotes value only elicitation.
\end{minipage}
\end{figure}
All comparisons use matched market seed cells. A cell fixes the market panel, retained securities, participant profile sequence, primitive shocks, contra side liquidity regime, and oracle benchmark. Unless otherwise noted, each reported mean uses 200 matched cells. Interval columns and figure intervals report the half width of a 95\% stratified bootstrap interval based on 9,999 replications. Pairwise claims use within cell differences, and the family of comparisons discussed within a table panel is adjusted by Holm's step down procedure. When no bootstrap draw exceeds the observed centered statistic, the reported result is \(p<0.001\). The compact labels used below are HY for hybrid demand and value elicitation, VO for value only elicitation, BR for bridge value query inclusion, NB for no bridge value query, SN for same securities crossing, NC for no completion, LR for low rank factor construction, and SH for shuffled factor construction.

\subsection{Search and Verification Architecture}
\label{secExpArchitecture}
This subsection tests the information architecture. The experiment asks whether welfare recovery comes from exact reports alone or from the combination of price directed search and exact verification. Figure \ref{figMainArchitecture} reports the two comparisons that determine the baseline architecture. The first panel fixes the total communication budget at 18 queries and varies the point at which the protocol switches from demand queries to value queries. The second panel varies the total query budget and compares hybrid elicitation with value only elicitation. Table \ref{tabBridgeAccounting} then isolates the bridge value query under tight value query budgets. The total query budget is a disclosure constraint rather than a tuning parameter. Figure \ref{figMainArchitecture} reports the full search-verification trade-off under the 18-query environment and shows that interior hybrid designs dominate the one-sided designs. We use the 12-demand-query and 6-value-query design as the baseline representative of this interior region; the substantive comparison is the complementarity between search and verification, not the optimality of a particular integer split. The bridge round is counted as one of the value-query rounds, so bridge and no-bridge designs face the same total communication budget.

Figure \ref{figMainArchitecture} shows that neither search nor verification is sufficient by itself. Value only elicitation obtains exact reports but recovers only 47.84\% of oracle welfare because it explores the high dimensional package space without price directed guidance. Demand only elicitation recovers 49.56\% because it locates active trading regions but relies on conservative lower bound accounting in \eqref{eqInferredValue}. The best interior split reaches 88.01\%, and the adjacent splits with 8 and 16 demand query rounds reach 86.43\% and 86.17\%. The pattern identifies a division of labor between the query types. Demand queries identify the portfolio regions in which mutually beneficial trades are likely to exist. Value queries provide the cardinal comparisons needed by the report based allocation rule. This distinction is consistent with the preference elicitation logic in combinatorial allocation, where demand and value queries provide different information about nonseparable values \citep{conen2001preference,brero2017probably,bichler2023fuel,soumalias2025prices}.

The scaling evidence in Figure \ref{figMainArchitecture} strengthens this interpretation. Hybrid elicitation rises from 67.60\% at 8 queries to 95.29\% at 48 queries. Value only elicitation rises from 42.22\% to 49.86\% over the same range. The hybrid advantage increases from 25.39 to 45.43\%age points, with a mean paired advantage of 37.90\%age points and Holm adjusted \(p<0.001\). Additional communication is valuable for the hybrid protocol because demand query moments improve the direction of surrogate search and value query observations anchor the welfare levels of the packages proposed by that search. Additional communication is less valuable for value only elicitation because exact values are attached to packages that are not guided by the price directed structure of portfolio demand.

Table \ref{tabBridgeAccounting} evaluates the empirical value of the bridge accounting logic. The bridge design exceeds the no bridge design by 5.61 to 7.03\%age points across lower bound fractions from 0.10 to 0.50. The mean paired advantage is 6.33\%age points with Holm adjusted \(p<0.001\). The bridge value query is useful because it changes the accounting status of the demand discovered allocation. Once that allocation is recorded as an exact value report, later value query exploration through \eqref{eqPredictedVQAllocation} can enlarge the report set while retaining the exact valued incumbent as a feasible option in the final report based allocation problem.

\begin{table}[t]
\caption{Bridge value query preservation under tight value query budgets}
\label{tabBridgeAccounting}
\centering
\footnotesize
\begin{tabular}{rrrr}
\toprule
DQ lower bound fraction & BR Eff. \(\%\) & NB Eff. \(\%\) & BR minus NB pp \\
\midrule
0.10 & 81.58 (1.05) & 75.97 (0.61) & 5.61 \\
0.20 & 81.82 (1.01) & 75.93 (0.61) & 5.89 \\
0.35 & 82.76 (1.15) & 75.96 (0.61) & 6.80 \\
0.50 & 83.01 (1.12) & 75.98 (0.61) & 7.03 \\
\bottomrule
\end{tabular}

\vspace{0.4em}
\begin{minipage}{0.92\linewidth}
\footnotesize
Notes. Parentheses report stratified bootstrap 95\% half widths. pp denotes\%age points. BR denotes bridge value query inclusion. NB denotes no bridge value query.
\end{minipage}
\end{table}
\subsection{Diagnostic Decomposition: Demand-Query Search and Report Accounting}
\label{secExpDiagnostic}

Table \ref{tabMainDqDiagnostic} examines why the hybrid architecture improves on demand only elicitation. The table evaluates demand discovered candidate packages under alternative accounting rules. The first row isolates the search component by asking how valuable the demand-discovered candidate set is before report-accounting restrictions are imposed. The remaining rows apply the report-based accounting rules available to the platform. The next rows use revealed preference lower bound accounting. The final rows add bridge valuation either within the same communication budget or with one additional query. The table is a mechanism diagnostic for the difference between locating valuable packages and using them as welfare evidence in the report based allocation rule.

\begin{table}[t]
\caption{Demand query search and valuation accounting}
\label{tabMainDqDiagnostic}
\centering
\footnotesize
\begin{tabular}{lrrr}
\toprule
Accounting specification & Total queries & Eff. \(\%\) & 95\% h.w. \\
\midrule
Primitive valuation of demand discovered candidates & 18 & 96.71 & 0.43 \\
Conservative revealed preference bound & 18 & 48.04 & 4.06 \\
Moderate revealed preference bound & 18 & 52.42 & 4.09 \\
High revealed preference bound & 18 & 56.67 & 4.00 \\
Full price revealed preference bound & 18 & 59.35 & 3.96 \\
Bridge valuation within fixed communication & 18 & 62.42 & 3.26 \\
Bridge valuation with expanded communication & 19 & 62.71 & 3.53 \\
\bottomrule
\end{tabular}
\vspace{0.4em}
\begin{minipage}{0.92\linewidth}
\footnotesize
Notes. The first row evaluates the same demand discovered candidate set at primitive values. The remaining rows use report based accounting rules that differ in how much of the revealed preference price information is admitted before exact value verification.
\end{minipage}
\end{table}

The diagnostic shows that demand queries are strong search instruments but incomplete welfare evidence. The primitive-value row asks how valuable the demand-discovered candidate set would be if the platform could evaluate those candidates directly. It recovers 96.71\% of full-information welfare, showing that demand queries locate economically relevant packages. The report-based rows then impose the information actually available to the platform. Under conservative revealed-preference accounting, the same candidate set supports only 48.04\%, and even the full-price revealed-preference bound supports 59.35\%. Incumbent verification raises the restricted diagnostic allocation to 62.42\% within the fixed communication budget.

This 62.42\% number should be read as an accounting diagnostic, not as the performance of the full hybrid protocol. The full hybrid result of 88.01\% in Figure \ref{figMainArchitecture} uses the demand-query phase to guide later model-guided value queries, which expand the exact-valued report set before final allocation. Table \ref{tabMainDqDiagnostic} instead holds the candidate set fixed to isolate why unverified demand responses cannot by themselves support welfare-comparable report-based allocation. The comparison therefore explains the mechanism behind the main result: search finds promising packages, while verification and subsequent value-query exploration make those packages usable for welfare allocation.

\subsection{Security Level Packages and Package Size Restrictions}
\label{secExpModes}

This subsection evaluates the exact securities operating mode before imposing disclosure penalties. The purpose is to distinguish SL packages from mechanical same securities crossing. SN restricts internal matching to direct offsetting interest in the same security. SL packages use demand discovered securities as package coordinates and then allocate through \eqref{eqReportWDP}. Table \ref{tabMainSecurity} compares crossing modes in the low leakage environment using unadjusted efficiency and then imposes minimum package size restrictions that mimic block style filters. Residual \(L_1\) measures the unmatched position after internal crossing, and crossed notional measures the absolute internal allocation size.

\begin{table}[t]
\caption{Security level packages and package size restrictions}
\label{tabMainSecurity}
\centering
\footnotesize
\begin{tabular}{lrrr}
\toprule
\multicolumn{4}{l}{\textit{Panel A. Crossing modes}} \\
Label & Eff. \(\%\) & Residual \(L_1\) & Crossed notional \\
\midrule
SL & 97.25 (0.86) & 0.60 & 4.11 \\
FC & 88.85 (0.86) & 0.96 & 3.45 \\
SN & 94.69 (0.55) & 0.00 & 3.67 \\
VO & 46.38 (1.41) & 0.69 & 2.46 \\
\midrule
\multicolumn{4}{l}{\textit{Panel B. Minimum package size restrictions}} \\
Label & Min. \(L_1\) & Eff. \(\%\) & Crossed notional \\
\midrule
SL & 0.00 & 97.25 (0.86) & 4.11 \\
FC & 0.00 & 88.85 (0.86) & 3.45 \\
SN & 0.00 & 94.69 (0.55) & 3.67 \\
SL & 0.25 & 97.25 (0.87) & 4.11 \\
FC & 0.25 & 88.85 (0.82) & 3.45 \\
SN & 0.25 & 94.69 (0.54) & 3.67 \\
SL & 0.50 & 96.71 (1.08) & 4.11 \\
FC & 0.50 & 88.71 (0.86) & 3.48 \\
SN & 0.50 & 94.69 (0.55) & 3.67 \\
SL & 0.75 & 96.12 (1.35) & 4.14 \\
FC & 0.75 & 86.03 (1.10) & 3.63 \\
SN & 0.75 & 84.35 (1.59) & 2.93 \\
\bottomrule
\end{tabular}

\vspace{0.4em}
\begin{minipage}{0.92\linewidth}
\footnotesize
Notes. Parentheses report stratified bootstrap 95\% half widths. SL denotes security level packages. FC denotes factor completed basket packages. SN denotes same securities crossing. VO denotes value only elicitation.
\end{minipage}
\end{table}

Panel A shows that SL packages are the unadjusted efficiency leader in the exact securities environment. SL reaches 97.25\% efficiency, compared with 94.69\% for SN, 88.85\% for FC, and 46.38\% for VO. The SL advantage over SN is 2.56\%age points with Holm adjusted \(p<0.001\). This comparison matters because SN has zero residual in the table but lower welfare than SL. Exact clearing of the same securities is not equivalent to efficient portfolio crossing. SL can accept residual external execution costs when those costs are dominated by participant surplus, as represented in \eqref{eqReportWDP}. This distinction is economically important because institutional portfolio value depends on covariance risk, benchmark deviation, liquidity curvature, and hedge completion rather than on security by security imbalance alone \citep{markowitz1952portfolio,grinold1999active,almgren2001optimal}.

Panel B shows that SL remains stable when the platform requires larger packages. SL moves from 97.25\% at a zero threshold to 96.12\% at a threshold of 0.75. FC moves from 88.85\% to 86.03\%. SN is more sensitive at the highest threshold and falls from 94.69\% to 84.35\%. The interpretation is consistent with the package family in \eqref{eqSecurityLevelFamily}. Same securities crossing depends on exact opposing pairs and loses welfare when small direct matches are removed. SL uses the same securities but embeds them in elicited packages, so it retains larger feasible combinations even when block style filters are imposed.

\subsection{Representation Frontier}
\label{secExpFrontier}

This subsection treats representation as an operating-mode decision under costly pretrade communication. Unadjusted efficiency is the relevant criterion when exact-security disclosure is inexpensive. Disclosure-adjusted welfare is the relevant criterion when pretrade messages reveal costly information about securities and trade direction.

Let \(N_a\) denote the active participant-security content of representation \(a\):
\begin{equation}
N_a = \frac{1}{nm} \sum_{i=1}^{n} \sum_{j=1}^{m} \mathbf 1\{|d_{ij}^{a,\mathrm{final}}|>\epsilon\}.
\end{equation}
Let \(L_a(\omega_a)=\min\{1,\omega_a N_a\}\), where \(\omega_a\) is the informativeness of a unit of message content under representation \(a\). Exact security-level disclosure is normalized to \(\omega_{\mathrm{SL}}=1\). Disclosure-adjusted welfare is
\begin{equation}
\operatorname{AdjEff}^{pp}_a(c) = 100\times \operatorname{Eff}_a - 100cL_a(\omega_a).
\end{equation}
For any two representations \(a\) and \(b\),
\begin{equation}
\operatorname{AdjEff}^{pp}_a(c)-\operatorname{AdjEff}^{pp}_b(c)
= 100\{\operatorname{Eff}_a-\operatorname{Eff}_b\}
- 100c\{L_a(\omega_a)-L_b(\omega_b)\}.
\label{eqAdjEffDifference}
\end{equation}
If representation \(a\) has lower unadjusted efficiency but lower message informativeness than representation \(b\), the break-even disclosure cost is
\begin{equation}
c^\star_{a\mid b} = \frac{ \operatorname{Eff}_b-\operatorname{Eff}_a
}{ L_b(\omega_b)-L_a(\omega_a)},
\label{eqRepresentationBreakEven}
\end{equation}
provided \(L_b(\omega_b)>L_a(\omega_a)\). Representation \(a\) lies above \(b\) on the disclosure-adjusted frontier when \(c\ge c^\star_{a\mid b}\).

The representation comparison therefore gives a frontier rather than a universal ranking. SL is the unadjusted-efficiency mode because exact-security packages preserve the allocation-relevant securities discovered by demand queries. FC gives up some exact-security flexibility, but it reduces the informativeness of pretrade package messages by expressing trades through factor-exposure targets and completed residual sleeves. The relevant design question is therefore not whether FC dominates SL in unadjusted efficiency. It is whether the platform's cost of pretrade message informativeness is high enough for the disclosure savings of FC to compensate for its unadjusted-efficiency loss. The break-even condition in \eqref{eqRepresentationBreakEven} characterizes this region directly. Figure \ref{figMainFrontier} plots one normalization of this frontier; Appendix Table \ref{tabAppLeakageBreakEven} reports the corresponding break-even informativeness weights under alternative leakage metrics.

\begin{figure}[t]
\centering
\includegraphics[width=0.68\linewidth]{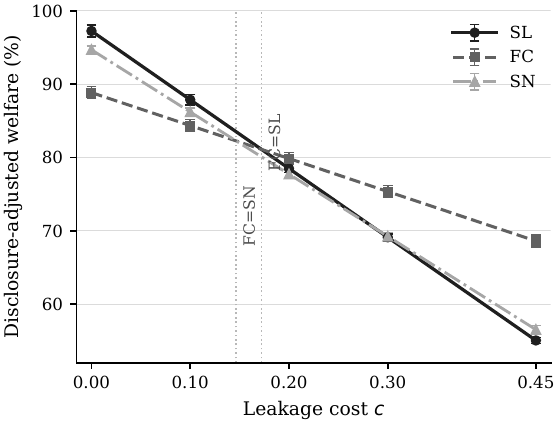}
\caption{Representation frontier under disclosure costs}
\label{figMainFrontier}
\vspace{0.25em}
\begin{minipage}{0.50\linewidth}
\footnotesize
Notes. Disclosure adjusted welfare equals \(100\times\operatorname{Eff}_a-100cL_a\). Vertical intervals are stratified bootstrap 95\% half widths. SL denotes security level packages. FC denotes factor completed basket packages. SN denotes same securities crossing.
\end{minipage}
\end{figure}

Figure \ref{figMainFrontier} shows a representation frontier rather than a universal ranking. At zero leakage cost, SL is highest at 97.25\%, followed by SN at 94.69\% and FC at 88.85\%. This ordering reflects the unadjusted efficiency cost of moving from exact securities packages to factor completed basket packages. When \(c=0.20\), FC becomes the disclosure adjusted welfare leader at 79.85\%, compared with 78.49\% for SL and 77.72\% for SN. At \(c=0.45\), FC reaches 68.60\%, while SL and SN reach 55.04\% and 56.52\%. Linear interpolation gives approximate crossover points of \(c\simeq0.172\) for FC relative to SL and \(c\simeq0.147\) for FC relative to SN. The result identifies the disclosure cost range in which the basket representation in \eqref{eqFactorCompletedFamily} becomes the disclosure adjusted welfare frontier.

\begin{table}[t]
\caption{Contra side liquidity and unadjusted representation efficiency}
\label{tabContraSideLiquidity}
\centering
\footnotesize
\begin{tabular}{lrrr}
\toprule
Label & \(\zeta_{\mathrm{contra}}\) & Eff. \(\%\) & Crossed notional \\
\midrule
SL & 0.25 & 72.58 (1.92) & 1.74 \\
FC & 0.25 & 74.47 (1.95) & 2.21 \\
SN & 0.25 & 52.81 (2.32) & 1.70 \\
SL & 0.50 & 82.69 (1.79) & 2.80 \\
FC & 0.50 & 81.38 (1.38) & 3.00 \\
SN & 0.50 & 78.16 (1.62) & 2.69 \\
SL & 0.75 & 92.26 (1.35) & 3.69 \\
FC & 0.75 & 86.31 (0.98) & 3.33 \\
SN & 0.75 & 92.02 (0.73) & 3.47 \\
SL & 1.00 & 95.26 (1.09) & 3.97 \\
FC & 1.00 & 88.63 (0.82) & 3.50 \\
SN & 1.00 & 94.58 (0.54) & 3.66 \\
\bottomrule
\end{tabular}

\vspace{0.4em}
\begin{minipage}{0.92\linewidth}
\footnotesize
Notes. Parentheses report stratified bootstrap 95\% half widths. SL denotes security level packages. FC denotes factor completed basket packages. SN denotes same securities crossing.
\end{minipage}
\end{table}

Table \ref{tabContraSideLiquidity} evaluates the same representations across market-thickness regimes, where thickness refers to the amount of direct opposing interest available to same-securities crossing. This comparison is not a tuning exercise; it identifies when exact-securities crossing is naturally favored and when portfolio packages create value by combining offsetting legs. In the thinnest contra side liquidity regime, SN recovers 52.81\%, while SL and FC recover 72.58\% and 74.47\%. At \(\zeta_{\mathrm{contra}}=0.50\), SL and FC are close, at 82.69\% and 81.38\%. At medium and high contra side liquidity, SL and SN lead in unadjusted efficiency. This pattern aligns with Proposition \ref{propPortfolioComb}. When exact same securities opposition is thin, value can still be created by portfolio packages whose legs offset each other through the curvature matrix \(H_i\). When opposing interest is thick and disclosure is inexpensive, exact securities representations remain the natural unadjusted efficiency modes.

\subsection{Robustness of the Factor Completed Basket Construction}
\label{secExpFamily}

This subsection evaluates whether factor completed basket packages are an internally coherent basket construction. The purpose is to test whether the completion map in \eqref{eqCompletion} improves on weaker basket alternatives under the same elicitation and allocation architecture. Table \ref{tabMainCompletionRobustness} has three roles. Panel A compares query family constructions. NC uses demand discovered securities without factor completion. LR uses a low rank factor basis. SH uses a shuffled factor basis that perturbs the economic content of the factor guide. Panel B varies the correlation regime through \(\Sigma_\alpha\). Panel C perturbs the surrogate inputs used for query selection while evaluating welfare under the true primitives.

\begin{table}[t]
\caption{Robustness of factor completed basket package construction}
\label{tabMainCompletionRobustness}
\centering
\footnotesize
\begin{tabular}{lrrr}
\toprule
\multicolumn{4}{l}{\textit{Panel A. Query family construction}} \\
Label & Eff. \(\%\) & 95\% h.w. & Rank \\
\midrule
SL & 95.60 & 0.99 & 1 \\
FC & 88.75 & 0.88 & 2 \\
LR & 85.20 & 1.05 & 3 \\
SH & 84.41 & 1.04 & 4 \\
NC & 81.92 & 1.13 & 5 \\
VO & 48.46 & 1.47 & 6 \\
\midrule
\multicolumn{4}{l}{\textit{Panel B. Factor completion across correlation regimes}} \\
Correlation \(\alpha\) & FC Eff. \(\%\) & NC Eff. \(\%\) & FC minus NC pp \\
\midrule
0.25 & 92.50 (0.85) & 86.32 (1.65) & 6.18 \\
0.65 & 90.03 (0.71) & 81.25 (1.32) & 8.78 \\
0.95 & 88.93 (0.90) & 81.43 (1.19) & 7.50 \\
\midrule
\multicolumn{4}{l}{\textit{Panel C. Surrogate input misspecification}} \\
Misspecification \(\mu\) & FC Eff. \(\%\) & NC Eff. \(\%\) & FC minus NC pp \\
\midrule
0.00 & 88.76 (0.82) & 81.48 (1.24) & 7.28 \\
0.10 & 88.64 (0.94) & 82.44 (1.20) & 6.20 \\
0.25 & 86.59 (0.91) & 81.66 (1.19) & 4.94 \\
0.50 & 82.39 (1.26) & 79.18 (1.48) & 3.22 \\
\bottomrule
\end{tabular}

\vspace{0.4em}
\begin{minipage}{0.92\linewidth}
\footnotesize
Notes. h.w. denotes stratified bootstrap 95\% half width. Parentheses in Panels B and C report the corresponding half widths. SL denotes security level packages. FC denotes factor completed basket packages. LR denotes low rank factor construction. SH denotes shuffled factor construction. NC denotes no completion. VO denotes value only elicitation.
\end{minipage}
\end{table}

Panel A shows that SL remains the unadjusted efficiency benchmark at 95.60\%. FC reaches 88.75\%, ahead of LR at 85.20\%, SH at 84.41\%, NC at 81.92\%, and VO at 48.46\%. The ordering supports the role of completion within the basket family. FC does not rely only on a lower dimensional factor summary. It combines the factor sleeve with a residual securities sleeve, which preserves participant specific active demand while controlling unintended factor exposure.

Panels B and C isolate this mechanism. Across correlation regimes, FC exceeds NC by 6.18 to 8.78\%age points. Under surrogate input misspecification, FC remains above NC at all reported levels, with a mean paired advantage of 5.41\%age points and Holm adjusted \(p<0.001\). The gap narrows at the highest misspecification level, which is consistent with the role of the surrogate in query selection. Since final allocation still uses \eqref{eqReportWDP}, misspecification affects which packages are queried rather than how elicited packages are valued. The evidence supports FC as a coherent basket representation relative to incomplete factor based alternatives.

The experiments support a search and verification interpretation of query based portfolio crossing. Demand queries identify active regions of signed portfolio demand, value queries verify selected packages, and the bridge value query records the demand discovered allocation as an exact valued candidate in the final report based allocation problem. Package representation is a distinct design margin. SL is the unadjusted efficiency mode when exact securities disclosure is inexpensive. FC is the basket mode that becomes preferred under the maintained disclosure adjusted welfare criterion when pretrade leakage costs are material and remains internally robust relative to incomplete basket constructions.

%% file: sections/chap5_conclusion.tex
\section{Conclusion}
\label{secConclusion}

This paper studies hidden liquidity discovery for institutional investors as a portfolio preference elicitation problem. Institutions evaluate signed portfolio trades rather than isolated single securities orders. The marginal value of a trade depends on the surrounding portfolio. A purchase can reduce welfare if it increases an existing exposure, while a sale can create value if it completes a hedge or reduces benchmark deviation. The object elicited by the platform is therefore a nonseparable valuation over signed portfolio trades.

Query based portfolio crossing separates search from verification. Demand queries use price directed responses to identify active regions of portfolio demand. Value queries record surplus for selected packages. The bridge value query adds exact values for the interim allocation selected after the demand query phase. The final allocation is chosen from the elicited report set and demand query lower bounds. The surrogate model selects queries and does not determine final welfare.

The formal analysis clarifies the accounting and representation logic behind this architecture. The quadratic surplus model provides a portfolio accounting identity for complementarities and substitutions across signed trade legs. The representation construction shows how security level and factor completed basket packages can be embedded in a common allocation problem. The report based allocation result clarifies why demand reports and value reports play different roles. Demand reports support search but enter final welfare comparisons through valid lower bounds unless verified. Value reports create exact welfare candidates. The bridge value query is the resulting incumbent verification step. Under valid lower bound accounting, later value query exploration retains the exact valued demand discovered allocation as a feasible welfare candidate in the final report based allocation problem.

The computational experiments support the search and verification interpretation. With an 18 query budget, value only and demand only designs recover 47.84\% and 49.56\% of oracle welfare, while the best hybrid split recovers 88.01\%. Across total query budgets, hybrid elicitation rises from 67.60\% to 95.29\%, while value only elicitation remains between 42.22\% and 49.86\%. A demand query diagnostic shows that demand discovered candidates can be highly valuable under primitive valuation, while conservative report accounting requires selective verification before those candidates can support final welfare comparisons.

The experiments also show that package representation is a design variable. Security level packages are the exact securities mode when pretrade securities disclosure is inexpensive. Factor completed basket packages are the portfolio basket mode. FC does not dominate exact securities packages in unadjusted efficiency. Instead, it becomes the disclosure adjusted welfare frontier once leakage costs are moderate or high. In thin contra side liquidity regimes, same securities crossing weakens sharply, while portfolio package representations retain welfare by combining legs that offset each other through portfolio curvature. The comparison is therefore a representation frontier rather than a universal ranking.

\appendix

%% file: appendix/Proofs.tex
\section{Proofs}
\label{appProofs}

\begin{proof}[Proof of Proposition \ref{propPortfolioComb}]
The reduced form valuation is \(v_i(d)=\theta_i^\top d-d^\top H_i d/2\). For two trade increments \(g\) and \(h\),
\begin{equation}
v_i(g+h)=
\theta_i^\top g+\theta_i^\top h-\frac{1}{2}g^\top H_i g-\frac{1}{2}h^\top H_i h
-g^\top H_i h .
\end{equation}
Subtracting \(v_i(g)\) and \(v_i(h)\) gives \(v_i(g+h)-v_i(g)-v_i(h)=-g^\top H_i h\). The sign of this cross term determines complementarity and substitutability.
\end{proof}

\begin{proof}[Proof of Lemma \ref{lemExactCandidatePreservation}]
Fix any \(a=(q_i)_{i\in\mathcal I}\in\mathcal A^{\mathrm{exact}}(\mathcal R)\). Since every selected package in \(a\) is exact valued,

\begin{equation}
\tilde v_i(q_i,\mathcal R_i)=v_i(q_i) \quad \text{for all }i\in\mathcal I.
\end{equation}
Because the residual execution cost \(\Psi\) is known and is evaluated in the same way in true and reported welfare,
\begin{equation}
\widetilde W(a\mid\mathcal R) =\sum_{i\in\mathcal I}\tilde v_i(q_i,\mathcal R_i) -\Psi\left(-\sum_{i\in\mathcal I}q_i\right) =W(a).
\end{equation}
By report based optimality,
\begin{equation}
\widetilde W(a^{\mathrm{final}}\mid\mathcal R) \ge \widetilde W(a\mid\mathcal R) = W(a).
\end{equation}
By lower bound validity,
\begin{equation}
\tilde v_i(q_i^{\mathrm{final}},\mathcal R_i) \le v_i(q_i^{\mathrm{final}})
\quad
\text{for all }i\in\mathcal I.
\end{equation}
Therefore,
\begin{equation}
\begin{aligned}
W(a^{\mathrm{final}}) &= \sum_{i\in\mathcal I}v_i(q_i^{\mathrm{final}})
- \Psi\left(-\sum_{i\in\mathcal I}q_i^{\mathrm{final}}\right)\\
&\ge \sum_{i\in\mathcal I}\tilde v_i(q_i^{\mathrm{final}},\mathcal R_i) - \Psi\left(-\sum_{i\in\mathcal I}q_i^{\mathrm{final}}\right)\\
&= \widetilde W(a^{\mathrm{final}}\mid\mathcal R).
\end{aligned}
\end{equation}
Combining the inequalities gives
\begin{equation}
W(a^{\mathrm{final}})
\ge \widetilde W(a^{\mathrm{final}}\mid\mathcal R) \ge W(a).
\end{equation}
The argument holds for every \(a\in\mathcal A^{\mathrm{exact}}(\mathcal R)\). Hence
\begin{equation}
W(a^{\mathrm{final}})
\ge
\max_{a\in\mathcal A^{\mathrm{exact}}(\mathcal R)} W(a).
\end{equation}

If the bridge value query records exact values for the interim demand query allocation \(a^{\mathrm{dq}}\), then \(a^{\mathrm{dq}}\in\mathcal A^{\mathrm{exact}}(\mathcal R)\) after the bridge reports are added. Applying the first part of the result gives
\begin{equation}
W(a^{\mathrm{final}}) \ge \max_{a\in\mathcal A^{\mathrm{exact}}(\mathcal R)} W(a) \ge W(a^{\mathrm{dq}}).
\end{equation}
\end{proof}

\begin{proof}[Proof of Proposition \ref{propDemandDualCertificate}]
First fix an arbitrary feasible allocation \(q=(q_i)_{i\in\mathcal I}\) with \(q_i\in\mathcal X_i\). Let
\begin{equation}
\xi_q=-\sum_{i\in\mathcal I}q_i
\end{equation}
be its residual execution vector. For any price vector \(p\), feasibility of this residual identity implies \(\xi_q+\sum_i q_i=0\). Hence
\begin{align}
\sum_i v_i(q_i)-\Psi(\xi_q) &=
\sum_i\{v_i(q_i)-p^\top q_i\}-\Psi(\xi_q)-p^\top\xi_q \\
&\le \sum_i U_i^v(p) + \sup_{\xi\in\mathbb R^m}\{-\Psi(\xi)-p^\top\xi\} \\
&= \sum_i U_i^v(p)+\Psi^*(-p) = \mathfrak C^v(p).
\end{align}
Taking the supremum over all feasible allocations gives weak duality,
\begin{equation}
W^\star\le \mathfrak C^v(p).
\end{equation}

Now let \(a\in\mathcal A^{\mathrm{exact}}(\mathcal R)\). Lemma \ref{lemExactCandidatePreservation} gives
\begin{equation}
W(a^{\mathrm{final}})\ge W(a).
\end{equation}
Combining this inequality with weak duality yields
\begin{equation}
W^\star-W(a^{\mathrm{final}})
\le
\mathfrak C^v(p)-W(a),
\end{equation}
which proves \eqref{eqGeneralDualCertificate}. If the bridge value query verifies \(a^{\mathrm{dq}}\), then \(a^{\mathrm{dq}}\in\mathcal A^{\mathrm{exact}}(\mathcal R)\), so substituting \(a=a^{\mathrm{dq}}\) proves \eqref{eqBridgeDualCertificate}.

It remains to specialize the bound to a simultaneous demand response at price \(p\). If \(d_i(p)\in\arg\max_{d\in\mathcal X_i}\{v_i(d)-p^\top d\}\), then
\begin{equation}
U_i^v(p)=v_i(d_i(p))-p^\top d_i(p).
\end{equation}
Let \(\xi(p)=-\sum_i d_i(p)\). Taking \(a=d(p)\) in \eqref{eqGeneralDualCertificate} gives
\begin{equation}
W^\star-W(a^{\mathrm{final}})
\le
\mathfrak C^v(p)-W(d(p)),
\end{equation}
and
\begin{align}
\mathfrak C^v(p)-W(d(p)) &= \sum_i\{v_i(d_i(p))-p^\top d_i(p)\}
+ \Psi^*(-p) - \left\{ \sum_i v_i(d_i(p))-\Psi(\xi(p)) \right\} \\
&= -p^\top\sum_i d_i(p)+\Psi^*(-p)+\Psi(\xi(p)) \\
&= \Psi(\xi(p))+\Psi^*(-p)+p^\top\xi(p).
\end{align}
This proves \eqref{eqDQDualGapGeneral}. The right hand side is nonnegative by Fenchel's inequality,
\begin{equation}
\Psi(\xi(p))+\Psi^*(-p)\ge (-p)^\top\xi(p).
\end{equation}

For the quadratic residual execution cost \(\Psi(\xi)=\xi^\top\Gamma\xi/2\), the convex conjugate is \(\Psi^*(z)=z^\top\Gamma^{-1}z/2\). Hence \(\Psi^*(-p)=p^\top\Gamma^{-1}p/2\), and
\begin{align}
\Psi(\xi(p))+\Psi^*(-p)+p^\top\xi(p)
&= \frac{1}{2}\xi(p)^\top\Gamma\xi(p) + \frac{1}{2}p^\top\Gamma^{-1}p
+p^\top\xi(p) \\
&= \frac{1}{2} \left(\xi(p)+\Gamma^{-1}p\right)^\top \Gamma \left(\xi(p)+\Gamma^{-1}p\right) \\
&= \frac{1}{2} \left( -\sum_i d_i(p)+\Gamma^{-1}p \right)^\top \Gamma \left( -\sum_i d_i(p)+\Gamma^{-1}p \right).
\end{align}
Combining this expression with \eqref{eqDQDualGapGeneral} proves \eqref{eqDQDualGapQuadratic}.
\end{proof}

%% file: appendix/detail.tex
\section{Experimental Calibration, Inference, and Robustness Details}
\label{appExpCalibration}

This appendix describes the empirical environment behind Section \ref{secExperiments}. The design holds the market environment fixed within matched cells and varies the elicitation architecture or the package representation. This structure ensures that differences in reported welfare come from the crossing procedure rather than from changes in the underlying economy.

\subsection{Market Environment}
\label{appExpData}

The experimental environment is constructed to preserve the allocation problem studied in Section \ref{secExperiments}. Portfolio crossing is valuable when participants face nonseparable values over signed trades, and this nonseparability depends on covariance, liquidity, and factor exposure. For this reason, the experiments use equity panels from the S\&P 500, KOSPI 200, Nikkei 225, and DAX rather than synthetic covariance matrices. Within each market, eligible securities are screened for return coverage and ranked by terminal market capitalization. The retained universe contains the largest liquid securities after this screen. This rule gives a comparable institutional trading universe across markets and avoids a design in which thinly traded tail securities drive the crossing outcome.

The baseline cell uses twenty securities, five factors, and eight participants. The market inputs are estimated from a one year daily return window. The retained return panel is used to estimate the covariance input \(\Sigma\). Returns are winsorized at the security level, the covariance guide is lightly regularized toward a diagonal structure, and the resulting matrix is projected to be positive semidefinite. Liquidity curvature is represented by \(\Delta=\operatorname{diag}(1/\ell_j)\), where \(\ell_j\) is a normalized market capitalization based liquidity measure. Factor exposures are the leading principal components of \(\Sigma\). The factor atom \(A\) and residual completion matrix \(R_K\) are then constructed as in \eqref{eqAtomMatrix} and \eqref{eqCompletion}.

These inputs are held fixed within matched market seed cells. A cell fixes the market panel, retained securities, participant profile sequence, primitive shocks, contra side liquidity, and oracle benchmark. The comparison between SL and FC therefore changes the package representation, not the market data supplied to the platform. The comparison between hybrid, value only, bridge, and no bridge designs likewise changes the elicitation architecture, not the underlying economy.

The communication environment is deliberately limited. The baseline uses eighteen queries per participant, and the query budget experiments vary this total to study how welfare recovery scales with communication. This restriction is part of the economic design rather than a computational convenience. A crossing platform must learn enough about portfolio values to allocate efficiently while limiting the amount of pretrade information revealed through messages.

This construction links the empirical design to Proposition \ref{propPortfolioComb}. The covariance matrix creates cross security curvature, the liquidity matrix creates name specific execution curvature, and the factor structure creates a low dimensional exposure space. Together, these features generate settings in which same name crossing, security level packages, and factor completed basket packages solve different allocation problems. Same name crossing can perform well when direct opposing interest is thick. SL can exploit exact name information when such disclosure is inexpensive. FC can preserve basket level information while reducing the informativeness of package messages.

\subsection{Participant Values and Feasible Trades}
\label{appParticipantCalibration}

The participant calibration creates institutional heterogeneity while preserving matched comparisons across protocols. Each matched cell contains eight participants drawn from five economic profiles. Indexers and ETF style participants generate benchmark and basket demand. Active managers generate sparse alpha motivated demand. Hedge accounts generate long short residual demand. Dealers generate inventory and liquidity intermediation demand. Table \ref{tabAppParticipantCalibration} reports the profile level parameters that affect feasible trades and curvature.

\begin{table}[t]
\caption{Participant primitive parameters}
\label{tabAppParticipantCalibration}
\centering
\footnotesize
\begin{tabular}{lrrrr}
\toprule
Profile \(g\) & \(\bar\lambda_g\) & \(\bar\gamma_g\) & Gross cap \(G_g\) & Name cap \(C_g\) \\
\midrule
Indexer & 2.60 & 0.36 & 1.30 & 0.16 \\
Active & 1.20 & 0.20 & 1.00 & 0.20 \\
Hedge & 2.10 & 0.18 & 1.15 & 0.22 \\
ETF & 1.40 & 0.30 & 1.35 & 0.18 \\
Dealer & 3.00 & 0.22 & 1.20 & 0.22 \\
\bottomrule
\end{tabular}

\vspace{0.4em}
\begin{minipage}{0.92\linewidth}
\footnotesize
Notes. Curvature parameters govern the relative importance of covariance risk and liquidity cost. Exposure caps define the feasible trade set in \eqref{eqAppFeasibleSet}.
\end{minipage}
\end{table}

For participant \(i\) with profile \(g_i\), feasible trades satisfy
\begin{equation}
\mathcal X_i
=
\left\{
 d\in\mathbb R^m
 \mid
 \|d\|_1\le G_{g_i},
 \ |d_j|\le C_{g_i}\text{ for }j=1,\ldots,m
\right\}.
\label{eqAppFeasibleSet}
\end{equation}
The no trade vector is feasible for every participant. Curvature is generated from the profile parameters and enters the value function through
\begin{equation}
H_i=\lambda_i\Sigma+\gamma_i\Delta+\rho I .
\label{eqAppCurvatureMatrix}
\end{equation}

Private target trades combine a factor component, a residual name component, and a profile specific trading motive. Contra side liquidity \(\zeta_{\mathrm{contra}}\) determines how much direct offsetting interest is present in the cell. Low values create thinner direct opposing interest and make direct same name matching more difficult. High values create richer opposing interest and make exact name crossing more favorable. The target trade is formed as
\begin{equation}
\tau_i
=
\Pi_{\mathcal X_i}\{Az_i+R_Ku_i\},
\label{eqAppTargetTrade}
\end{equation}
where \(z_i\) is a factor level demand shock, \(u_i\) is a sparse residual name shock, and \(\Pi_{\mathcal X_i}\) projects the candidate trade into the feasible set. Active managers receive sparse alpha shocks, while dealers receive inventory pressure shocks. Private values are generated by
\begin{equation}
\theta_i=H_i\tau_i+\alpha_i,
\qquad
v_i(d)=\theta_i^\top d-\frac{1}{2}d^\top H_id .
\label{eqAppValueGeneration}
\end{equation}

This structure creates two sources of allocation value. Common factor components create broad crossing opportunities, while residual name components create participant specific package value. The design therefore gives same name crossing, security level packages, and factor completed basket packages a common comparison environment. Same name crossing can perform well when direct opposition is thick. SL can exploit demand discovered securities when exact name disclosure is acceptable. FC can preserve basket level information while reducing the informativeness of messages.

\subsection{Query Construction and Report Accounting}
\label{appQueryConstruction}

The query design follows Algorithm \ref{algDemandDiscovered}. Demand queries are used first because they locate active regions of the portfolio demand surface. Value queries are then used because the final winner determination problem requires cardinal welfare comparisons. The default split uses the 18 query environment studied in Figure \ref{figMainArchitecture}. The bridge value query is counted as the first value query round, which makes the bridge comparison a communication matched comparison rather than an additional information grant.

Demand query reports are treated as search information unless they are later verified by a value query. If participant \(i\) chooses \(d_i^\ell\) at query price \(p^\ell\), revealed preference against no trade gives a conservative value bound. The recorded bound is
\begin{equation}
\underline v_i^\ell
=
\begin{cases}
\rho_{\mathrm{dq}}{p^\ell}^\top d_i^\ell & \text{if }{p^\ell}^\top d_i^\ell\ge0\\
{p^\ell}^\top d_i^\ell & \text{if }{p^\ell}^\top d_i^\ell<0 .
\end{cases}
\label{eqAppDQLowerBound}
\end{equation}
This rule lets demand reports guide search without allowing unverified demand responses to enter the final allocation as exact welfare observations. Value query reports record surplus, and the no trade package is included with value zero for every participant.

Demand query prices are updated in the projected dual space described in \eqref{eqTimeVaryingPriceBasis} and \eqref{eqPriceUpdate}. Early rounds use factor and liquidity directions. Later rounds augment the basis with securities that appear repeatedly in demand responses. Participant specific allocation relevant securities are selected from cumulative report activity. Let
\begin{equation}
M_{ij}=\sum_{r\in\mathcal R_i}|x_{rj}|
\label{eqAppsecuritiescores}
\end{equation}
denote the activity score for name \(j\). The largest coordinates of \(M_i\) define \(S_i\). SL uses \(E_{S_i}\) as the package basis. FC uses the factor sleeve and the residual name sleeve through \([A,R_KE_{S_i}]\). NC removes completion, LR uses a lower rank factor guide, and SH perturbs the economic content of the factor guide. These variants are used only to identify the role of factor completion in Table \ref{tabMainCompletionRobustness}.

The surrogate model is estimated from two types of observations. Demand query rows fit projected marginal conditions, while value query rows fit exact surplus observations. The estimator is used only to select later queries. Final welfare is evaluated from the report based winner determination problem in \eqref{eqReportWDP}. This distinction matters for the interpretation of the robustness results. Misspecification affects which packages are queried, but it does not redefine the welfare criterion once packages have been reported.

\subsection{Evaluation and Statistical Inference}
\label{appComputationInference}

The evaluation uses matched market seed cells. A cell fixes the market panel, retained universe, participant profile sequence, primitive shocks, contra side liquidity, and oracle benchmark. Each protocol is then run on the same cell. This design makes pairwise comparisons within cell comparisons and removes variation that would otherwise come from drawing different economies.

The welfare benchmark is the full information allocation in \eqref{eqOracle}. The implemented allocation is the report based allocation in \eqref{eqReportWDP}. Efficiency is reported as the oracle relative welfare measure in \eqref{eqEfficiency}. Candidate reports include no trade, demand query reports, the bridge value query report, and later value query reports. Efficiency ratios are reported for cells with a positive oracle benchmark.

Uncertainty intervals are computed by stratified bootstrap resampling over matched cells. Pairwise claims use within cell differences. When a table discusses several pairwise comparisons within the same panel, the reported inference uses Holm step down adjustment. This inference design matches the experimental question. The relevant comparison is whether one protocol improves welfare when it faces the same market, participants, and primitives as the alternative protocol.

\subsection{Pretrade Information Leakage}
\label{appLeakageRobustness}

The representation frontier in Figure \ref{figMainFrontier} treats pretrade disclosure as costly. The main text uses an active name pair leakage index because it is directly tied to the information contained in package messages. For representation \(a\), the index is
\begin{equation}
L_a^{\mathrm{active}}
=
\min
\left\{
1,
\omega_a
\frac{1}{nm}
\sum_{i=1}^{n}
\sum_{j=1}^{m}
\mathbf 1\{|d_{ij}^{a,\mathrm{final}}|>\epsilon\}
\right\}.
\label{eqAppLeakageMain}
\end{equation}
The multiplier \(\omega_a\) captures representation specific message informativeness. SL discloses exact name packages. SN also discloses exact securities but with a narrower same name structure. FC expresses packages through factor exposure and completed residual sleeves, so its messages are less directly informative about exact name intentions. Disclosure adjusted welfare is
\begin{equation}
\operatorname{AdjEff}^{pp}_a(c,L)
=
100\operatorname{Eff}_a-100cL .
\label{eqAppAdjustedEfficiency}
\end{equation}

Table \ref{eqRepresentationBreakEven} reports the robustness check that supports the frontier interpretation. Panel A changes the basis of the leakage index while holding the baseline informativeness multipliers fixed. Panel B changes the FC multiplier while holding the unadjusted welfare outcomes fixed. The table shows the range of disclosure environments in which FC becomes the disclosure adjusted welfare frontier because its lower message informativeness offsets its lower unadjusted efficiency.

\begin{table}[t]
\caption{Break-even message informativeness for the basket frontier}
\label{eqRepresentationBreakEven}
\centering
\footnotesize
\begin{tabular}{lrrrr}
\toprule
\multicolumn{5}{l}{\textit{Panel A. Alternative leakage metrics at \(c=0.45\)}} \\
Leakage metric & SL disc. adj. welfare \(\%\) & FC disc. adj. welfare \(\%\) & SN disc. adj. welfare \(\%\) & Frontier \\
\midrule
Active pair & 55.04 & 68.60 & 56.52 & FC \\
Effective name & 64.89 & 75.06 & 67.24 & FC \\
Quantity weighted & 55.17 & 71.98 & 65.08 & FC \\
Top name & 52.25 & 68.60 & 56.44 & FC \\
\midrule
\multicolumn{5}{l}{\textit{Panel B. Break-even FC informativeness at \(c=0.45\)}} \\
Leakage metric & FC multiplier & SL disc. adj. welfare \(\%\) & FC disc. adj. welfare \(\%\) & SN disc. adj. welfare \(\%\) \\
\midrule
Active pair & 0.30 & 55.04 & 75.35 & 56.52 \\
Active pair & 0.45 & 55.04 & 68.60 & 56.52 \\
Active pair & 0.60 & 55.04 & 61.85 & 56.52 \\
Active pair & 0.75 & 55.04 & 55.10 & 56.52 \\
Effective name & 0.30 & 64.89 & 79.66 & 67.24 \\
Effective name & 0.45 & 64.89 & 75.06 & 67.24 \\
Effective name & 0.60 & 64.89 & 70.46 & 67.24 \\
Effective name & 0.75 & 64.89 & 65.86 & 67.24 \\
Quantity weighted & 0.30 & 55.17 & 77.60 & 65.08 \\
Quantity weighted & 0.45 & 55.17 & 71.98 & 65.08 \\
Quantity weighted & 0.60 & 55.17 & 66.36 & 65.08 \\
Quantity weighted & 0.75 & 55.17 & 60.73 & 65.08 \\
Top name & 0.30 & 52.25 & 75.35 & 56.44 \\
Top name & 0.45 & 52.25 & 68.60 & 56.44 \\
Top name & 0.60 & 52.25 & 61.85 & 56.44 \\
Top name & 0.75 & 52.25 & 55.10 & 56.44 \\
\bottomrule
\end{tabular}
\vspace{0.4em}
\begin{minipage}{0.92\linewidth}
\footnotesize
Notes. Panel A holds representation specific informativeness multipliers fixed. Panel B varies only the FC multiplier.
\end{minipage}
\end{table}

\subsection{Approximate Demand Responses}
\label{appDemandResponseRobustness}

Table \ref{tabAppDemandResponseRobustness} relaxes the exact optimizer assumption in the demand query step. The result is stable across moderate response approximation. SL remains the stronger unadjusted efficiency representation, while FC remains close enough to support the representation frontier in environments with material disclosure costs. The welfare object remains the allocation obtained from limited reports under a fixed information architecture.

\begin{table}[t]
\caption{Robustness to approximate demand responses}
\label{tabAppDemandResponseRobustness}
\centering
\footnotesize
\begin{tabular}{lrrr}
\toprule
Demand response specification & Label & Eff. \(\%\) & 95\% h.w. \\
\midrule
Exact optimizing response & SL & 96.75 & 0.95 \\
Exact optimizing response & FC & 88.16 & 0.74 \\
Mild stochastic response error & SL & 96.89 & 0.83 \\
Mild stochastic response error & FC & 89.21 & 0.76 \\
Moderate stochastic response error & SL & 95.74 & 0.74 \\
Moderate stochastic response error & FC & 89.63 & 0.79 \\
Large stochastic response error & SL & 89.30 & 0.78 \\
Large stochastic response error & FC & 86.79 & 0.73 \\
Mild regularized response & SL & 96.53 & 1.08 \\
Mild regularized response & FC & 88.40 & 0.67 \\
Moderate regularized response & SL & 97.20 & 0.89 \\
Moderate regularized response & FC & 89.36 & 0.75 \\
Strong regularized response & SL & 96.88 & 0.92 \\
Strong regularized response & FC & 89.93 & 0.69 \\
\bottomrule
\end{tabular}

\vspace{0.4em}
\begin{minipage}{0.92\linewidth}
\footnotesize
Notes. The approximation families preserve matched market seed cells and evaluate final welfare under primitive values. The stochastic response rows perturb demand choices around the exact optimizing response. The regularized response rows smooth the demand choice problem before reports are passed to the same allocation rule.
\end{minipage}
\end{table}